\documentclass[sigconf]{acmart}

\usepackage{graphicx}
\usepackage[ruled,linesnumbered,vlined]{algorithm2e} 
\usepackage{varwidth}
\usepackage{algpseudocode}
\usepackage{amssymb}
\usepackage{mathrsfs} 
\usepackage{array}
\usepackage{mdwmath}
\usepackage{mdwtab}
\usepackage{eqparbox}
\usepackage{tabularx}
\usepackage[font=footnotesize,caption=false]{subfig} 
\usepackage[belowskip=0pt,aboveskip=5pt]{caption}
\usepackage{stfloats}
\usepackage{url}
\usepackage{booktabs}
\usepackage{hyperref}    
\usepackage{amsfonts}
\usepackage{amsmath}
\usepackage{makecell}

\setcopyright{rightsretained}

\newcommand{\ourimpl}{\textsc{Pgbsc}}
\newcommand{\fasciaimpl}{\textsc{Fascia}}
\newcommand{\pvtsc}{\textsc{PFascia}}

\acmConference[HPDC'19]{International ACM Symposium on
High-Performance Parallel and Distributed Computing
}{June 2019}{Phoenix, Arizona, USA}
\acmYear{2019}
\copyrightyear{2019}
\begin{document}
\title{A GraphBLAS Approach for Subgraph Counting}

\author{Langshi Chen}
\affiliation{%
  \institution{Indiana University}
}
\email{lc37@indiana.edu}

\author{Jiayu Li}
\affiliation{%
  \institution{Indiana University}
}
\email{jl145@iu.edu}

\author{Ariful Azad}
\affiliation{%
  \institution{Indiana University}
}
\email{azad@iu.edu}

\author{Lei Jiang}
\affiliation{%
  \institution{Indiana University}
}
\email{jiang60@iu.edu}

\author{Madhav Marathe}
\affiliation{%
 \institution{University of Virginia}
}
\email{marathe@virginia.edu}

\author{Anil Vullikanti}
\affiliation{%
  \institution{University of Virginia}
}
\email{vsakumar@virginia.edu}

\author{Andrey Nikolaev}
\affiliation{%
  \institution{Intel Corporation}
  }
\email{andrey.nikolaev@intel.com}

\author{Egor Smirnov}
\affiliation{%
  \institution{Intel Corporation}
  }
\email{egor.smirnov@intel.com}

\author{Ruslan Israfilov}
\affiliation{%
  \institution{Intel Corporation}
  }
\email{ruslan.israfilov@intel.com}

\author{Judy Qiu}
\affiliation{%
  \institution{Indiana University}
  }
\email{xqiu@indiana.edu}

\renewcommand{\shortauthors}{Chen et al.}

\begin{abstract}
Subgraph counting aims to count the occurrences of a subgraph template T in a given network G. The basic problem of computing structural properties such as counting triangles and other subgraphs has found applications in diverse domains. Recent biological, social, cybersecurity and sensor network applications have motivated solving such problems on massive networks with billions of vertices. 
The larger subgraph problem is known to be memory bounded and computationally challenging to scale; the complexity grows both as a function of T and G. 
In this paper, we study the non-induced tree subgraph counting problem, propose a novel layered software-hardware co-design approach, and implement 
a shared-memory multi-threaded algorithm: 1) reducing the complexity of the parallel color-coding algorithm by identifying and pruning
redundant graph traversal; 2) achieving a fully-vectorized implementation upon linear algebra kernels inspired by GraphBLAS, 
which significantly improves cache usage and maximizes memory bandwidth utilization. Experiments show that our implementation improves the 
overall performance over the state-of-the-art work by orders of magnitude and up to 660x for subgraph templates with size over 12 on a dual-socket Intel(R) Xeon(R) Platinum 8160 server. We believe our approach using GraphBLAS with optimized sparse linear algebra can be applied to other massive subgraph counting problems and emerging high-memory bandwidth hardware architectures.

\end{abstract}
%


\maketitle

\section{Introduction}
\label{sec:introduction}

The naive algorithm of counting the exact number of subgraphs (including a triangle in its simplest form) of size $k$ in a $n$-vertex network takes $O(k^{2}n^k)$ time and is an NP-hard problem and computationally challenging, even for moderate values of $n$ and $k$.
Nevertheless, counting subgraphs from a large network is fundamental in 
numerous applications, and approximate algorithms are developed to estimate 
the exact count with statistical guarantees. In protein research, the physical contacts between proteins in the cell are represented as a network, and this protein-protein interaction network (PPIN) is crucial in understanding the cell physiology which helps develop new drugs. 
Large PPINs may contain hundreds of thousands of vertices (proteins) and millions of edges (interactions) while they usually contain repeated local structures (motifs). 
Finding and counting these motifs (subgraphs) is essential to compare different PPINs. Arvind et al.\cite{ArvindApproximationAlgorithmsParameterized2002a} counts bounded treewidth graphs but still has a time complexity super-polynomial to network size $n$. Alon et al. \cite{AlonBiomolecularNetworkMotif2008, AlonBalancedFamiliesPerfect2007} provide a practical algorithm to count trees and graphs 
of bounded treewidth (size less than 10) from PPINs of unicellular and 
multicellular organisms by using the color-coding technique developed in~\cite{alon_color-coding_1995}. 
In online social network (OSN) analysis, the graph size could even reach a 
billion or trillion, where a motif may not be as simple as a vertex (user) with a high degree but a group of users sharing specific interests. Studying these groups improves the design of the OSN system and the searching algorithm. \cite{chen:icdm16} enables an estimate of graphlet (size up to 5) counts in an
OSN with 50 million of vertices and 200 million of edges. 
Although the color-coding algorithm in~\cite{AlonBiomolecularNetworkMotif2008}
has a time complexity linear in network size, it is exponential in the size of subgraph. Therefore, efficient parallel implementations are the only viable 
way to count subgraphs from large-scale networks. To the best of our knowledge, a multi-threaded implementation named FASCIA~\cite{SlotaFastApproximateSubgraph2013} is considered to be the state-of-the-art work in this area. Still, it takes FASCIA more than 4 days (105 hours) to count a 17-vertex subgraph from the RMAT-1M network (1M vertices, 200M edges) on a 48-core Intel (R) Skylake processor. While our proposed shared memory multi-threaded algorithm named \textbf{\ourimpl{}}, which prunes
the color-coding algorithm as well as implements GraphBLAS inspired 
vectorization, takes only 9.5 minutes to complete the same task on the same hardware.
The primary contributions of this paper are as follows:
\begin{itemize}
\item {\bf Algorithmic optimization.} We identify and reduce the redundant computation complexity of the sequential color-coding algorithm to improve the parallel performance by a factor of up to 86.
\item {\bf System design and optimization.} We refactor the data structure as well as the execution order to maximize the hardware
efficiency in terms of vector register and memory bandwidth usage.
The new design replaces the vertex-programming model by using linear algebra kernels inspired by the GraphBLAS approach.  
\item {\bf Performance evaluation and comparison to prior work.} We characterize the improvement compared to state-of-the-art work \fasciaimpl{} from both of theoretical analysis and experiment results, and our solution attains the full hardware efficiency according to a roofline model analysis. 
\end{itemize}

\section{Background and Related Work}
\label{sec:related work}
\subsection{Subgraph Counting by Color coding}
\label{sub:subgraph_counting}
A subgraph $H(V_H,E_H)$ of a simple unweighted graph $G(V,E)$ is a graph $H$ whose vertex set and edge set are subsets of those of $G$. 
$H$ is an embedding of a template graph $T(V_T,E_T)$ if $T$ is isomorphic to $H$.
The subgraph counting problem is to count the number of all embeddings of a given template $T$ in a network $G$. \par
\begin{figure}[ht]
    \centering
    \includegraphics[width=\linewidth]{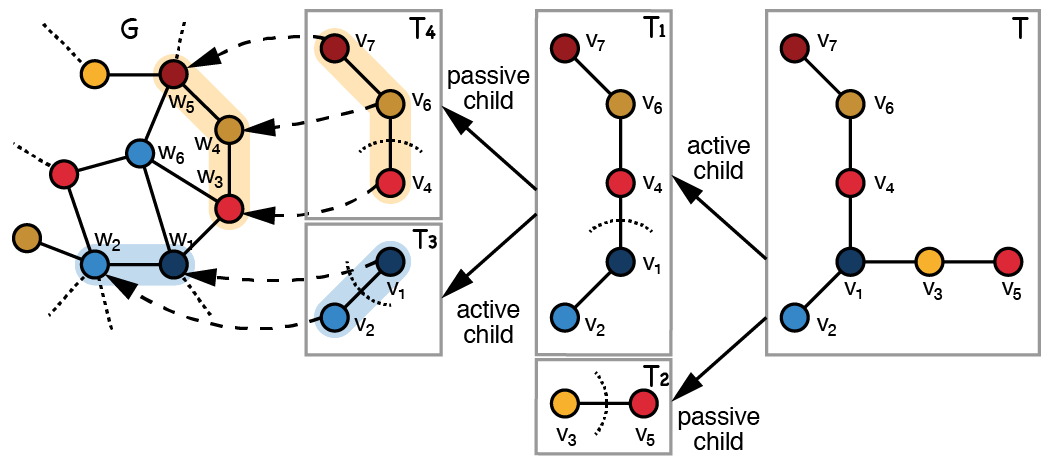}
    \caption{Illustration of the template partitioning within a
    colored input $G=(V,E)$}
    \label{fig:partitionTemplates}
\end{figure}
Color coding \cite{AlonBiomolecularNetworkMotif2008} provides a \emph{fixed parameter tractable} algorithm to address the subgraph counting problem where $T$ is a tree.
It has a time complexity of $O(c^k\text{poly}(n))$, which
is exponential to template size $k$ but polynomial to vertex number $n$. The algorithm consists of three main phases:
\begin{enumerate}
    \item \textbf{Random coloring}. Each vertex $v \in G(V,E)$
    is given an integer value of color randomly selected
    between $0$ and $k-1$, where $k \ge |V_T|$. 
    We consider $k=|T|$ for simplicity, and $G(V,E)$ is
    therefore converted to a labeled graph.
We consider an embedding $H$ as "colorful" if
each of its vertices has a distinct color value. In~\cite{alon_color-coding_1995}, Alon proves that the 
probability of $H$ being colorful is $\frac{k!}{k^k}$, 
and color coding approximates the exact number of $H$ by using the count of colorful $H$. 
\item \textbf{Template partitioning}.
When partitioning a template $T$, a single vertex is selected as the root $\rho$ while $T_s(\rho)$ refers to the $s$-th sub-template rooted at $\rho$. Secondly, one of the edges $(\rho, \tau)$ adjacent to root $\rho$ is cut, creating two child sub-templates. The child holding $\rho$ as its root is named \emph{active child} and denoted as $T_{s,a}(\rho)$. The child rooted at $\tau$ of the cutting edge is named \emph{passive child} and denoted as $T_{s,p}(\tau)$.
This partitioning recursively applies until each sub-template has just one vertex. A dynamic programming process is then applied in a bottom-up way through all the
sub-templates $T_s$ to obtain the count of $T$.

\item \textbf{Counting by dynamic programming}.
For each vertex $V_i \in G(V,E)$ at each sub-template $T_s$, we notate $N(V_i, T_s, C_s)$ as the count of 
embeddings of $T_s$ with its root mapped to $V_i$ using a color set $C_s$ drawn from $k=|T|$ colors.
Each $C_s$ is split into a color set $C_{s,a}$ for $T_{s,a}$ and another $C_{s,p}$ for $T_{s,p}$. For bottom sub-template $|T_s|=1$, $N(V_i, T_s, C_s)$ is $1$ only if $C_s$ equals the color randomly assigned to $V_0$ and otherwise it is $0$. For non-bottom sub-template with $|T_s|\ge1$, we have $N(V_i,T_s,C_s)=\sum_{V_j \in N(V_i)}N(V_i,T_{s,a},C_{s,a})N(V_j, T_{s,p},C_{s,p})$, where $N(V_i, T_{s,a}, C_{s,a})$ and $N(V_j, T_{s,p}, C_{s,p})$ have been calculated in previous steps of the dynamic programming because $|T_{s,a}|\le |T_s|, |T_{s,p}|\le |T_s|$. 
\end{enumerate}

A tree subgraph enumeration algorithm by combining color coding with a stream-based cover decomposition was developed in ~\cite{zhao2010subgraph}. To process massive networks, \cite{zhao_sahad:_2012} developed a distributed
version of color-coding based tree counting solution upon MapReduce framework in Hadoop, \cite{SlotaComplexNetworkAnalysis2014a} implemented a MPI-based solution, and ~\cite{zhao2018finding}~\cite{chen2018high} pushed the limit of
subgraph counting to process billion-edged networks 
and trees up to 15 vertices.\par

Beyond counting trees, a sampling and random-walk based technique has been applied to count graphlets, a small induced graph with size up to 4 or 5, which include the work of \cite{AhmedEfficientGraphletCounting2015} and \cite{chen:icdm16}. Later, \cite{chakaravarthy2016subgraph} extends 
color coding to count any graph with a treewidth of 2 in a distributed system. Also, \cite{RezaPruneJuicePruningTrillionedge2018} provides a pruning method on labeled networks and graphlets to reduce the vertex 
number by orders of magnitude prior to the actual counting. \par 

Other subgraph topics include: 1) \emph{subgraph finding}.  
As in~\cite{EkanayakeMIDASMultilinearDetection2018}, paths and trees with size up to 18 could be detected by using multilinear detection; 2) \emph{Graphlet Frequency Distribution} estimates relative frequency among all subgraphs with the same size
\cite{bressan_counting_2017}~\cite{PrzuljBiologicalNetworkComparison2010}; 3)
\emph{clustering} networks by using the relative frequency of their subgraphs
~\cite{RahmanGraftEfficientGraphlet2014}.
 
\subsection{GraphBLAS}
\label{subsec:graphblas}

A Graph can be represented as its adjacency matrix, 
and many key graph algorithms are expressed in terms of
linear algebra
~\cite{kepner2016mathematical}~\cite{kepner2015graphs}. 
The GraphBLAS project\footnote{\url{www.graphblas.org}} was 
inspired by the Basic Linear Algebra Subprograms (BLAS)
familiar to the HPC community 
with the goal of building graph algorithms upon a small 
set of kernels such as sparse matrix-vector 
multiplication (SpMV) and sparse matrix-matrix
multiplication (SpGEMM). The GraphBLAS community
standardizes~\cite{mattson2013standards} such kernel operations, 
and a GraphBLAS C API has been provided~\cite{mattson2017graphblas}. Systems consistent with GraphBLAS philosophy include: 
Combinatorial BLAS~\cite{BulucCombinatorialBLASdesign2011}, GraphMat~\cite{SundaramGraphMatHighPerformance2015},  Graphhulo~\cite{HutchisonNoSQLAccumuloNewSQL2016}, 
and GraphBLAS Template Library~\cite{ZhangGBTLCUDAGraphAlgorithms2016a}. 
Recently, SuiteSparse GraphBLAS~\cite{DavisGraphalgorithmsSuiteSparse2018} provides a sequential implementation of the GraphBLAS C API. GraphBLAST~\cite{YangDesignPrinciplesSparse2018a} provides
a group of graph primitives implemented on GPU.\par
GraphBLAS operations have been successfully employed to implement a suite of traditional graph algorithms 
including Breadth-first traversal
(BFS)~\cite{SundaramGraphMatHighPerformance2015},
Single-source shortest path
(SSSP)~\cite{SundaramGraphMatHighPerformance2015}, 
Triangle Counting~\cite{AzadParallelTriangleCounting2015},
and so forth. More complex algorithms have also been
developed with GraphBLAS primitives. For example, the
high-performance Markov clustering algorithm
(HipMCL)~\cite{azad2018hipmcl} that is used to cluster
large-scale protein-similarity networks is centered 
around a distributed-memory SpGEMM algorithm.

\begin{figure}
    \centering
    \includegraphics[width=\linewidth]{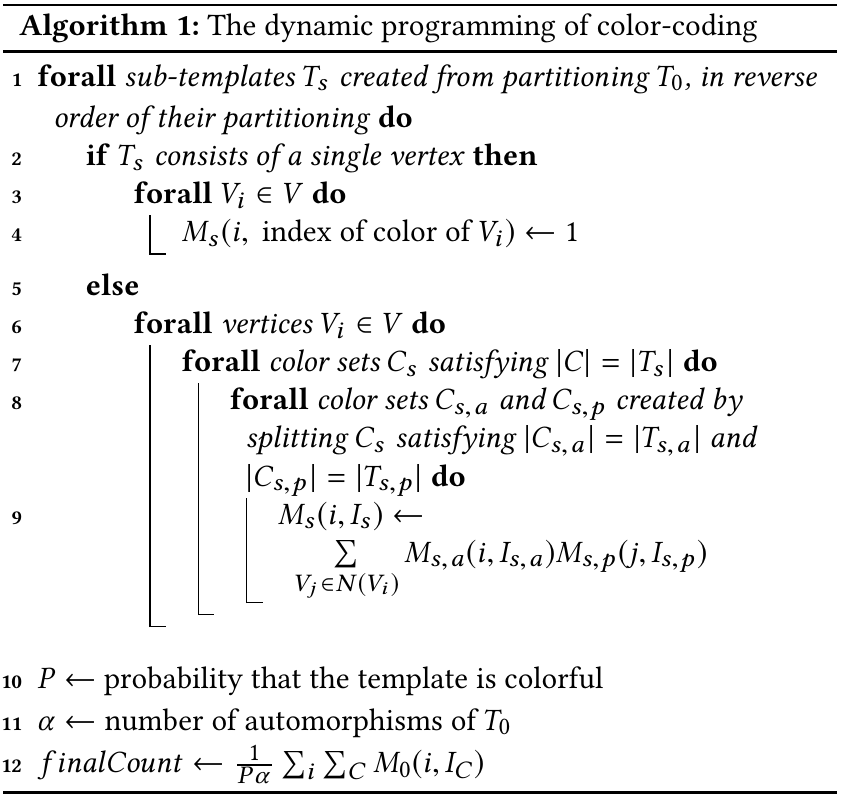}
\end{figure}
\setcounter{figure}{1}

\section{Color coding in Shared-Memory System}
\label{sec:parallel_colorcoding}
Algorithm 1 introduces the implementation of a single iteration of a color-coding algorithm within a multi-threading and shared-memory system, which is adopted by state-of-the-art work like \fasciaimpl{} ~\cite{SlotaFastApproximateSubgraph2013}. We refer to it as the \emph{\fasciaimpl{}} color coding for the rest of the paper, and we will address its performance issues from both algorithmic level and system implementation level. \par 
\begin{figure}[ht]
    \centering
    \includegraphics[width=\linewidth]{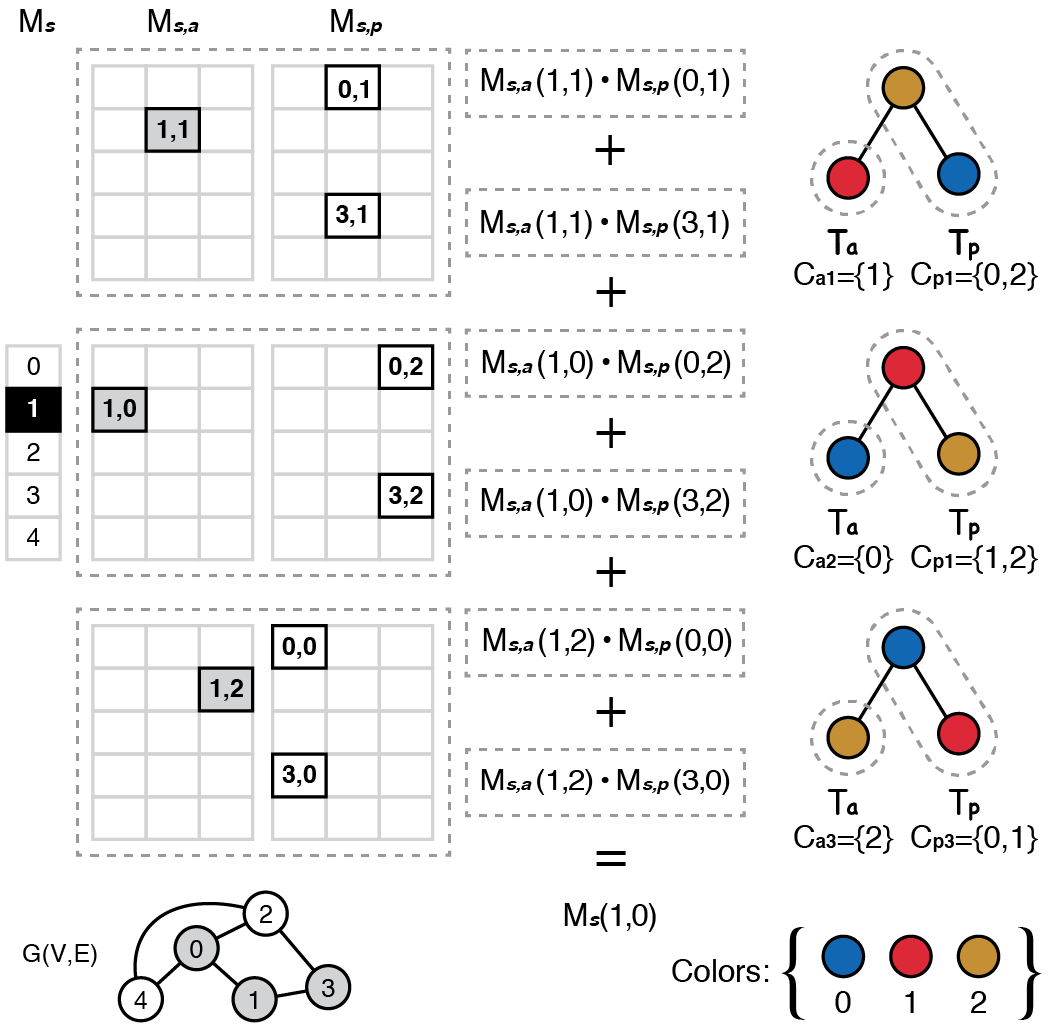}
    \caption{Illustrate Line 7 to 13 of Algorithm 1) by
    a five-vertex $G=(V,E)$ and a three-vertex template $T$. The column indices in $M_{s,a}$ and $M_{s,p}$ are
    calculated from their color combinations by Equation~\ref{eq:index}.}
    \label{fig:illustration-baseline-steps}
\end{figure}
A first \emph{for} loop over $T_s$ at line 1 implements the dynamic programming technique introduced in Section~\ref{sub:subgraph_counting},  
a second \emph{for} loop over all
$v_i\in G=(V,E)$ at line 6 is parallelized by threads. Hence, each thread processes the workload associated with one vertex at a time, and the workload of each thread includes another three \emph{for} loops: 
1) line 7 is a \emph{for} loop over all the color combination occurrences $C_s$, satisfying $|C_s| = |T_s|$; 
2) line 8 describes a loop over all the color set splits $C_{s,a}$ and $C_{s,p}$, where $|C_{s,a}|$, $|C_{s,p}|$ equals the sizes of partitioned active and passive children of template $T_s$ in Section~\ref{sub:subgraph_counting}; 
3) line 9 loops over all of the neighbor vertices of $V_i$ to multiply $N(V_i, T_{s,a}, C_{s,a})$ by $N(V_j,T_{s,p},C_{s,p})$. 
The \fasciaimpl{} color coding uses data structures 
as follows: 1) using an adjacency list $Adj[][]$ to hold the vertex IDs of each $V_i$'s neighbors, i.e., 
$Adj[i][]$ stores $V_j \in N(V_i)$; 2) Using a map to 
record all the sub-templates $T_s$ partitioned from $T$ in pairs $(s,T_s)$.
3) Using an array of dense matrix $M_s$ to hold the 
count data for each sub-template $T_s$. We have: 
\begin{itemize}
    \item Row $i$ of $M_s$ stores count data for each $C_s$ associated to a certain $V_i$
    \item Column $j$ in $M_s$ stores count data of all $V_i$ associated to a certain color combination $C_s$. 
    \item $M_s$ has $|V|$ rows and $\binom{|T|}{|T_s|}$ columns
\end{itemize}
We suppose that $T_s$ has an active 
child $T_{s,a}$ and a passive child $T_{s,p}$.
To generate an integer index for each color set $C_s$ in $M_s$, \cite{SlotaFastApproximateSubgraph2013} proposes an 
index system that hashes an arbitrary color set $C$ with arbitrary size at line 7 of Algorithm 1 into an unique 32-bit integer index according to 
Equation~\ref{eq:index}.  
\begin{equation}
\label{eq:index}
    I_c = \binom{c_1}{1} + \binom{c_2}{2} + \cdots + \binom{c_h}{h}
\end{equation}
Here, we have $h=|C|$ integers (color values) satisfying 
$0\leq c_1 <c_2 <\dots<c_h\leq k-1$, where $h\leq k$. 
We illustrate the major steps of
Algorithm 1 in Figure~\ref{fig:illustration-baseline-steps}, where a vertex $V_1$ is updating its count value from three color combinations for $T_{s,a}$ and $T_{s,p}$. Equation~\ref{eq:index}
calculates the column index value for
$T_{s,a}$ to access the its data in $M_{s,a}$ and $T_{s,p}$ to access the its data in $M_{s,p}$, respectively. Since $V_1$ has two neighbors of
$V_0, V_3$, for each $C_s$, it requires $M_{s,a}(1, I_{s,a})$ to multiply $M_{s,p}(0,I_{s,p})$ and $M_{s,p}(2,I_{s,p})$ accordingly.

\begin{table}[ht]
    \centering
    \caption{Definitions and Notations for Color coding}
    \label{tab:notation}
\resizebox{\linewidth}{!}{%
    \begin{tabular}{cc}
    \toprule
     Notation & Definition  \\
    \midrule
    $G(V,E)$, $T$ & Network and template \\ 
    $n$,$k$ & $n=|V|$ is vertex number, $k=|T|$ is template size \\
    $A_G$ & Adjacency matrix of $G(V,E)$ \\
    $T_s$, $T_{s,a}$, $T_{s,p}$ & The $s$th sub-template, active child of $T_s$, passive child of $T_s$ \\
    $C_s$, $C_{s,a}$, $C_{s,p}$ & Color set for $T_s$, $T_{s,a}$, $T_{s,p}$ \\ 
    $N(V_i,T_s,C_s)$ & Count of $T_s$ colored by $C_s$ on $V_i$ \\ 
    $M_s$, $M_{s,a}$, $M_{s,p}$ & Dense matrix to store counts for $T_s$, $T_{s,a}$, $T_{s,p}$ \\
    $I_s$, $I_{s,a}$, $I_{s,p}$ & Column index of color set
    $C_s$, $C_{s,a}$, $C_{s,p}$ in $M_s$, $M_{s,a}$, $M_{s,p}$ \\ 
    \bottomrule
    \end{tabular}
}
\end{table}

\subsection{Redundancy in Traversing Neighbor Vertices}
\label{sub:redundancy-nb-looping}

The first performance issue we observed in Algorithm 1 exists at the two-fold $\emph{for}$ loops from line 7 to line 9. When the sub-template size $|T_s| \le |T|$, it is probable to have multiple color combinations, where the passive children hold the same color set $C_{s,p}$ (i.e., the same column index) in $M_{s,p}$ while the active children have different color sets $C_{s,a1}, C_{s,a2}, \dots, C_{s,al}$, where $l=\binom{|T|-|T_{s,p}|}{|T_s|-|T_{s,p}|}$. 
Except for the last step of dynamic programming where $|T_s|=|T|$ and $l=1$, the repeated access to $M_{s,p}(:,I_{s,p})$ is considered to be redundant. 
\begin{figure}[ht]
    \centering
    \includegraphics[width=\linewidth]{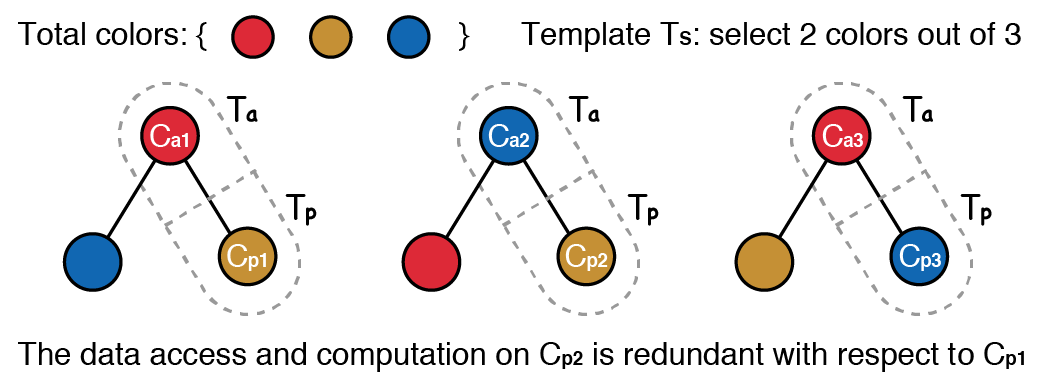}
    \caption{Identify the redundancy of \fasciaimpl{} color coding in a two-vertex sub-template $T_s$,
    which is further split into an active child and a passive child.}
    \label{fig:baseline-redundancy}
\end{figure}

In Figure~\ref{fig:baseline-redundancy}, we illustrate this redundancy by a $T_s$ with two colors taken out of three. Supposing $T_s$ has an active child $T_{s,a}$ with one color and a passive child $T_{s,p}$ with one color. The data access to $C_{s,p2}$ is redundant with respect to that of $C_{s,p1}$ because they share the same color of green.
This redundant data access on passive children is non-trivial because it loops over all neighbors of $V_i$, where data locality is usually poor. 
Therefore, we develop a pruning technique in Section~\ref{subs:pruning_combs} to address this issue.
 
\subsection{Lack of Vectorization and Data Locality}
\label{sub:lack_vec_locality}

Although $M_s$ stores its values in a dense matrix, the index system of Equation~\ref{eq:index} does not guarantee that the looping over $C_{s,a}$ and $C_{s,p}$ at line 8 of Algorithm 1 would have contiguous column indices at $M_{s,a}$ and $M_{s,p}$.
For instance, in Figure~\ref{fig:illustration-baseline-steps}, $I_{s,ca1}=1$, $I_{s,ca1}=0$ while $I_{s,ca3}=2$. 
Hence, the work by a thread on each row of $M_{s,a}$ and $M_{s,p}$ cannot be vectorized because of this irregular access pattern.
Also, this irregular access pattern compromises the data locality of cache usage. 
For example, a cache line that prefetches a 64-Byte chunk from memory address starting at $\&M_{s,a}(1,1)$ and ending at $\&M_{s,a}(1,16)$ cannot serve the request to access $M_{s,a}(1,0)$. 
It is even harder to cache data access to $M_{s,p}$ because they belong to different rows (neighbor vertices). We shall address this issue by re-designing data structure and thread execution order in Section~\ref{sub:precompute_neighbor} to \ref{sub:vec-thd-workflow}.

\section{Design and Implementation}
\label{sec:complana_matgb}

To resolve the performance issues in Section~\ref{sub:redundancy-nb-looping} and \ref{sub:lack_vec_locality}, we first identify and prune the unnecessary computation and memory access in Algorithm 1. Secondly, we modify the underlying data structure, the thread execution workflow, and provide a fully-vectorized implementation by using linear algebra kernels. 

\subsection{Pruning Color Combinations}
\label{subs:pruning_combs}
To remove the redundancy observed in Section~\ref{sub:redundancy-nb-looping}, we first apply the transformation by Equation~\ref{eq:transform_traversal} using the distributive property of addition and multiplication,
\begin{equation}
\label{eq:transform_traversal}
    \sum_{V_j \in N(i)}M_{s,a}(i, I_{s,a})\cdot M_{s,p}(j, I_{s,p}) = M_{s,a}(i, I_{s,a})\sum_{V_j \in N(i)}
    M_{s,p}(j, I_{s,p})
\end{equation}
 where it adds up all the $M_{s,p}(j,I_{s,p})$ before 
multiplying $M_{s,a}(i, I_{s,a})$ while keeping the same arithmetic result. The implementation will be illustrated in detail in Figure~\ref{fig:prunFascia}.
\begin{figure}[ht]
    \centering
    \includegraphics[width=\linewidth]{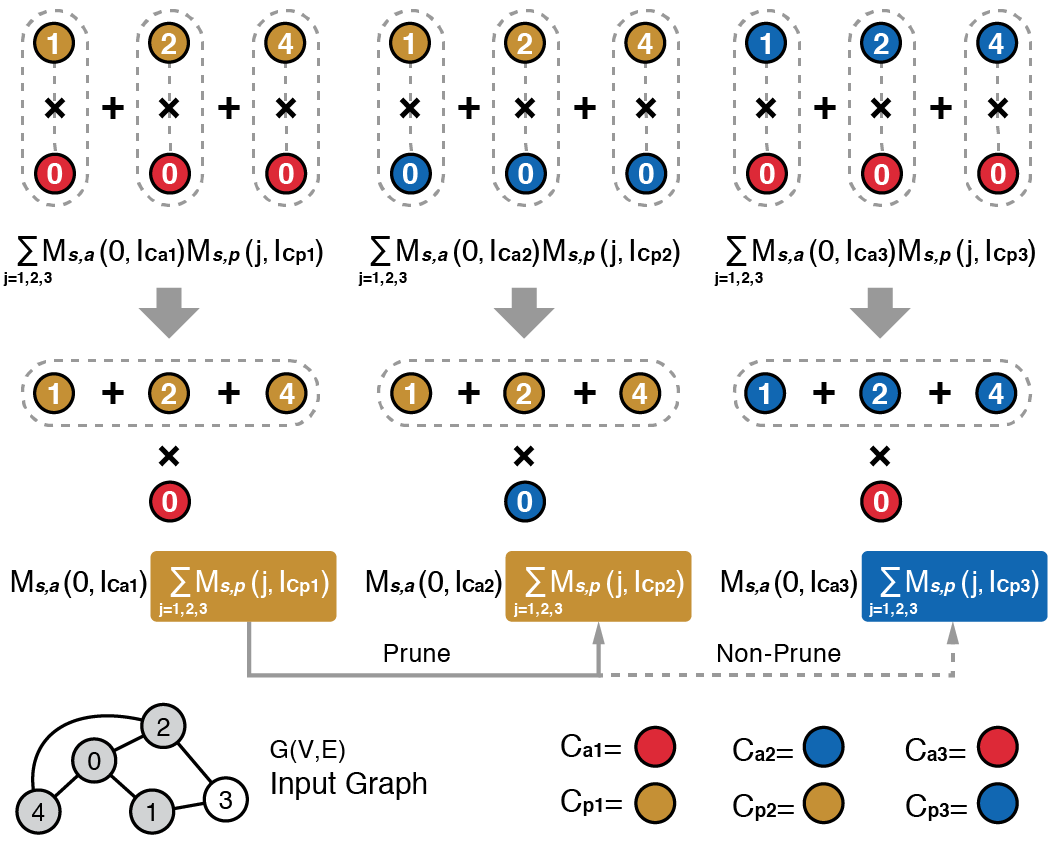}
    \caption{Illustrate steps of pruning optimization 
    from vertex 0 in $G=(V,E)$ with the three color
    combinations shown in Figure~\ref{fig:baseline-redundancy}. 1) Re-order addition and multiplication (grey arrow); 2): prune the vertex neighbor summation. 
    }
    \label{fig:illustration-pruning-redundancy}
\end{figure}

 Secondly, we 
check whether multiple color combinations share the same 
$C_{s,p}$ and prune them by re-using the result from its 
first occurrence. In Figure~\ref{fig:illustration-pruning-redundancy}, we examine a case with a sub-template $T_s$ choosing two colors out of three. In this case, 
$T_s$ is split into an active child $T_{s,a}$ with one vertex and a passive child $T_{s,p}$ with one vertex. Obviously, the neighbor traversal 
of vertex $V_0$ for $C_{s,a2},C_{s,p2}$ is pruned by 
using the results from $\{C_{s,a1}, C_{s,p1}\}$.

\subsection{Pre-compute Neighbor Traversal}
\label{sub:precompute_neighbor}
\begin{figure*}[ht]
    \centering
    \includegraphics[width=\textwidth]{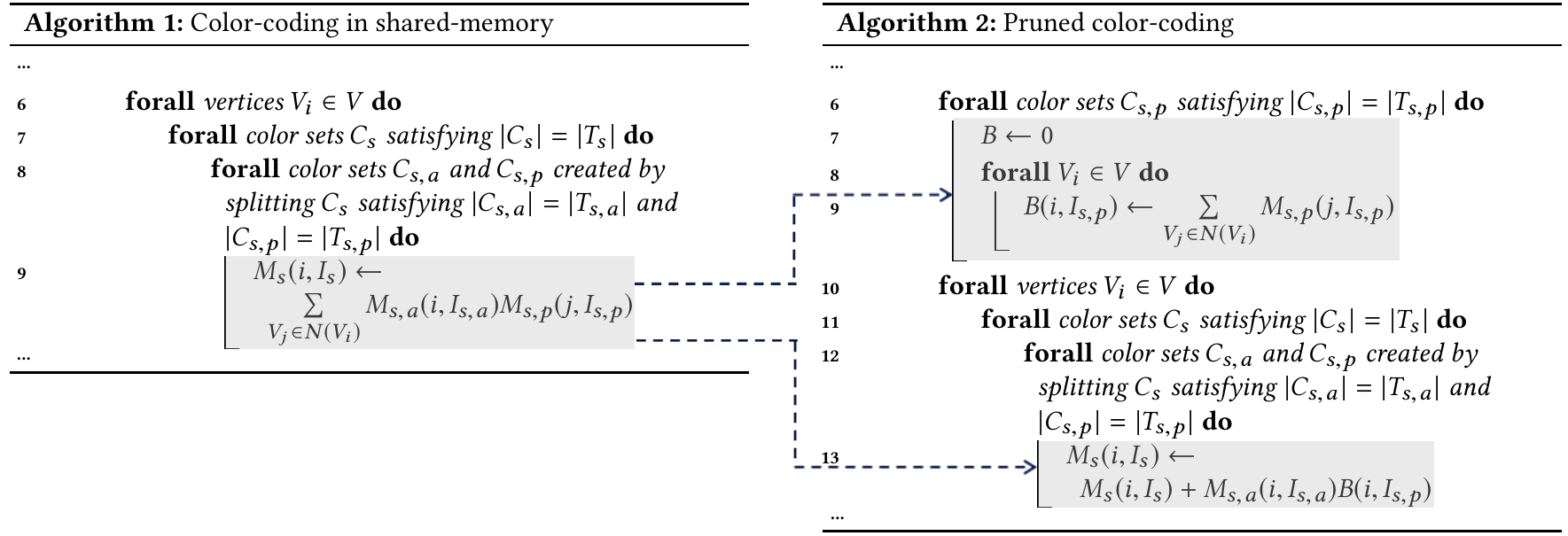}
    \caption{Pre-compute the pruned vertex-neighbor traversal by: 1) stripping out the summation of count data $M_{s,p}$ from its multiplication by count data from $M_{s,a}$; 2) storing the summation value of $M_{s,p}$ 
    and looking it up before multiplying count data from $M_{s,a}$}
    \label{fig:prunFascia}
\end{figure*}

Furthermore, the traversal of neighbors to sum up $M_{s,p}$ counts is stripped out from the 
\emph{for} loop at line 2 of Algorithm 1 as a pre-computation module shown in Figure~\ref{fig:illustration-pruning-redundancy}. We use an array buffer of length $|V|$ to temporarily hold the summation results for a certain $C_{s,p}$ across all $V_i$ 
with $SumBuf[i] = \sum_{V_j \in N(i)}M_{s,p}(j, I_{s,p})$. After the pre-computation for $C_{s,p}$, 
we write the content of buffer back to $M_{s,p}(V_i, I_{s,p})$ and release $SumBuf$ for the next $C_{s,p}$. We then replace the calculations of $\sum_{V_j \in N(i)}M_{s,p}(j, I_{s,p})$ by looking up their values from $M_{s,p}(i, I_{s,p})$. Meanwhile, except for the $|V|$-sized
array buffer, we do not increase the memory footprint 
for this pre-computation module.\par

It is worth noting that the pruning of vertex neighbor traversal shall also benefit the performance in a distributed memory system, where the data request to a vertex neighbor would be much expensive than that in a shared-memory system because of explicit interprocess communication.

\subsection{Data Structure for Better Locality}
\label{sub:new-data-structure}
\begin{figure*}[ht]
    \centering
    \includegraphics[width=\linewidth]{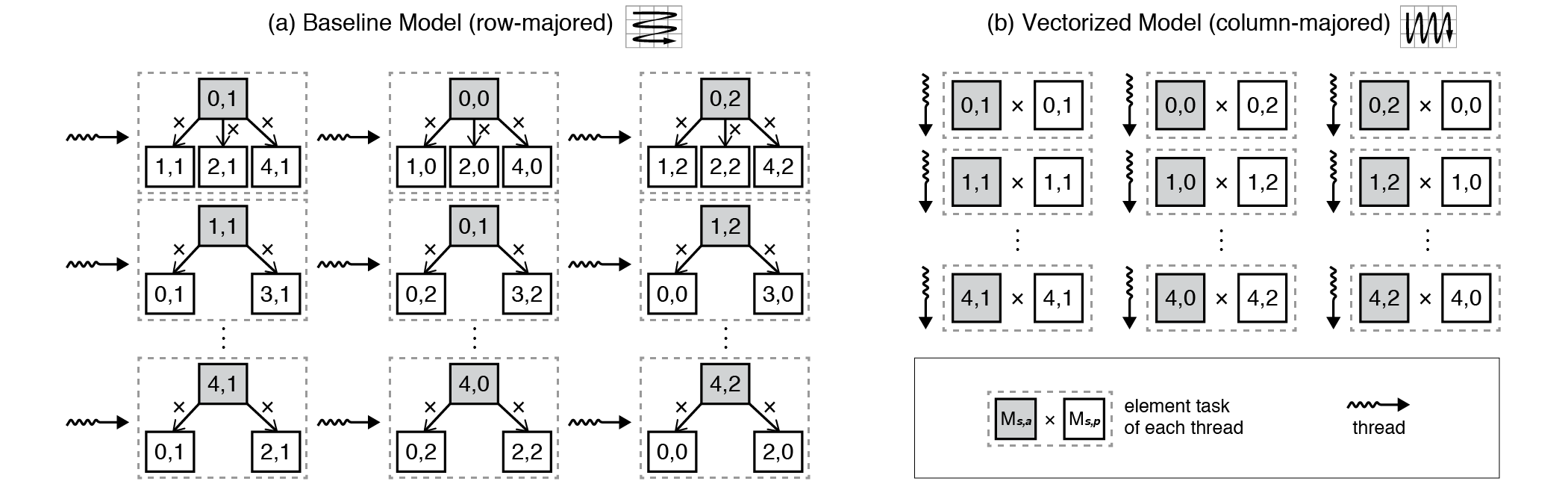}
    \caption{Comparing the thread execution order, where
    (a) \fasciaimpl{} color coding (baseline Model) has count data stored in memory with a row-majored layout, and (b) \ourimpl{} (Vectorized Model) has count data stored in memory with a column-majored
    layout.}
    \label{fig:fig:countTableOrder}
    \label{fig:thd-execution-order}
\end{figure*}

We replace the adjacency list $adj[][]$ of network $G(V,E)$ by using an adjacency matrix 
$A_G$, where $A_G(i,j) = 1$ if $V_j$ is a 
neighbor of $V_i$, and otherwise $A_G(i,j) = 0$. The
adjacency matrix has two advantages in our case: 
\begin{enumerate}
    \item It enables us to represent the neighbor vertex traversal at line 9 of Algorithm 2 by using a matrix operation.
    \item It allows a pre-processing step upon the sparse matrix to improve the data locality of the network if necessary. 
\end{enumerate}
The pre-processing step would be quite efficient in processing an extremely large network with a high sparsity in a distributed-memory system. For example, the Reverse Cuthill-McKee Algorithm (RCM) reduces the communication overhead as well as improves the bandwidth utilization~\cite{AzadReverseCuthillMcKeeAlgorithm2017}. Because the sparse matrix would be re-used by each vertex $V_i$, this additional pre-processing overhead is amortized and neglected in practice. \par

For count tables in $M[]$, we keep the dense matrices but
change the data layout in physical memory from row-majored order in Figure~\ref{fig:thd-execution-order}(a) to column-majored order in Figure~\ref{fig:thd-execution-order}(b). 
The column-majored viewpoint reflects a vector form of count data, i.e., the count data for a certain color combination $C$ is stored in a $|V|$ dimensional vector, which is essential to our next vectorization effort, because all of the count data for a certain color combination $C$ from all vertices $V_i \in G(V,E)$ are now stored in contiguous memory space. 

\subsection{Vectorized Thread Execution and Memory Access}
\label{sub:vec-thd-workflow}

With the new locality-friendly data structure, we vectorize the counting tasks and memory access by re-ordering the thread execution workflow. Here the 
thread is created by users, e.g., we use OpenMP 4.0 to spawn threads.\par

To vectorize the pre-computation of $\sum_{V_j \in N(i)}M_{s,p}(j,w)$ for each vertex $V_i$, we 
extend the buffer in Section~\ref{sub:redundancy-nb-looping} from an $|V|$ dimensional array into a
$|V|\times Z$ matrix $M_{buf}$, where $Z$ is a pre-selected batch size and data is stored in a row-majored order. We have the following procedure:
\begin{enumerate}
    \item For each vertex $V_i$, load the first $Z$ values from row $M_{s,p}(i,:)$ into $M_{buf}(i,:)$. 
    \item All rows (vertices) of $M_{buf}$ are processed by threads in parallel. 
    \item For each row, a thread sequentially loops over $V_j \in N(i)$ to calculate 
    $M_{s,p}(i,z)\gets M_{s,p}(i,z) + M_{buf}(j,z)$.
\end{enumerate}
The batch size value $Z$ shall be a multiple of the SIMD unit length or the cache line size, which 
ensures full utilization of the vector register length and the prefetched data in a cache line. \par

To vectorize the execution of count task at line 9 of Algorithm 1, we change the thread workflow from Figure~\ref{fig:thd-execution-order}(a) to Figure~\ref{fig:thd-execution-order}(b). 
In Figure~\ref{fig:thd-execution-order}(a), the \fasciaimpl{} color coding has each thread process counting
work belonging to a vertex at a time.  
Conversely, in Figure~\ref{fig:thd-execution-order}(b), we characterize the thread workflow as:
\begin{enumerate}
    \item Tasks belong to one color combination but for all vertices in $G(V,E)$, is dispatched altogether to all threads at a time. All threads collaboratively finish these tasks before moving to tasks for the next color combination.
    \item For the tasks from the same color combination, each thread is assigned an equal portion of element tasks with consecutive row numbers in $M_{s,a}$ and $M_{s,p}$.
    \item For each thread, within its portion of element tasks, the calculation is vectorized due to the consecutive row number and the column-majored data storage layout. 
\end{enumerate}
Typically $|V|$ has more than one million vertices in the network and is much larger than the total thread number. Each thread gets a portion of tasks much larger than the SIMD unit length and ensures full utilization of its SIMD resource. Furthermore, regular stride 1 memory access of each thread on $M_{s,a}$ and $M_{s,p}$ enables an efficient data prefetching from memory into cache lines.

\subsection{Exploit Linear Algebra Kernels}
\label{sub:kernel-spmm-ema}
\setcounter{algocf}{2}

\begin{figure}[ht]
\includegraphics[width=\linewidth]{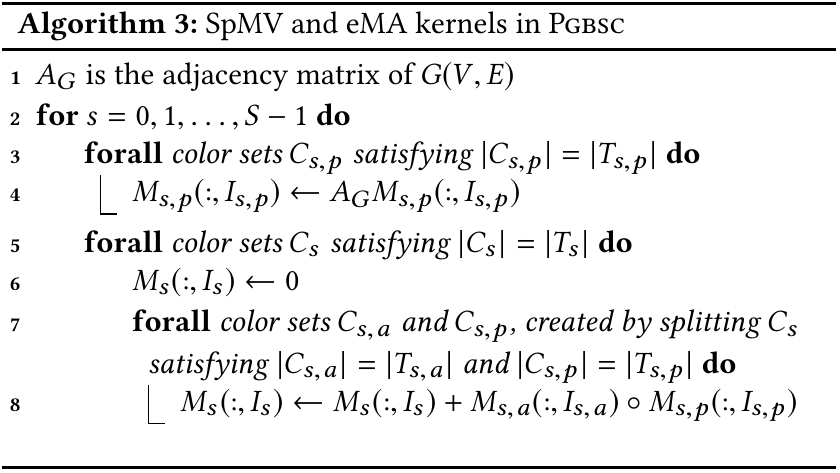}
\end{figure}

According to the GraphBLAS philosophy, the looping over
a vertex's neighbors is represented as a multiplication of a vector by the adjacency matrix. 
Therefore, the looping of neighbors for each $V_i \in V$ for
pre-computing $M_{s,p}$ at line 8,9 of Algorithm 2 is written as $A_G M_{s,p}(:,I_{s,p})$ at line 4 of Algorithm 3, where $A_G$ is the adjacency
matrix notated in Section~\ref{sub:new-data-structure} and
$I_{s,p}$ is the column index for a color set $C_{s,p}$ in $M_{s,p}$. Such multiplication is identified as Sparse Matrix-Vector multiplication (SpMV). 
To save memory space, we store the multiplication
result back to $M_{s,p}(:,I_{s,p})$ with the help of a buffer.
In practice, our \ourimpl{} adopts a Compressed Sparse Column (CSC) format on $A_G$ and the looping over color sets
$C_{s,p}$ to do SpMV at line 3,4 in Algorithm 3 is combined with a sparse matrix dense matrix (SpMM) 
kernel to apply the vectorization in Algorithm 4. \par

\begin{figure}[ht]
\includegraphics[width=\linewidth]{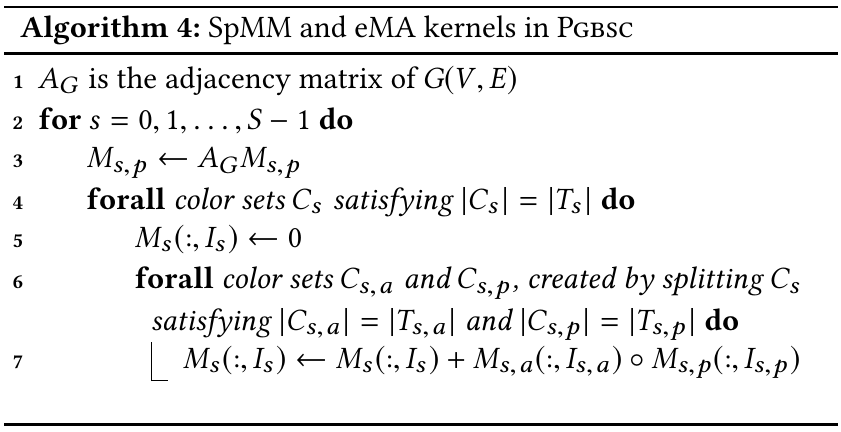}
\end{figure}

%
%
%

The vectorized counting task is implemented by an element-wise vector-vector multiplication and addition (named as eMA) kernel at line 7 of Algorithm 4, where vector $M_s(:,I_{c})$, $M_{s,a}(:,I_{s,a})$, and $M_{s,p}(:,I_{s,p})$
are retrieved from $M_s$, $M_{s,a}$, and $M_{s,p}$, respectively. 
We code this kernel by using Intel (R) AVX intrinsics, where multiplication and addition are implemented by using fused multiply-add (FMA) instruction set that leverages the 256/512-bit vector registers.

\section{Complexity Analysis}
\label{sec:compl-analysis}

In this section, we will first analyze the time complexity of the \ourimpl{} algorithm to show a significant reduction. Then, we will give the formula to estimate the improvement of \ourimpl{} versus \fasciaimpl{}.

\subsection{Time Complexity}
\label{subsec:compl-three}

We assume that the size of the sub-templates is uniformly distributed between $1$ and $n$. We will use two matrix/vector operations, sparse matrix dense matrix multiplication (SpMM) and element-wise multiplication and addition (eMA). The time complexity of SpMV is considered as $|E|$, and the time complexity of eMA is $|V|$.

\begin{lemma}
For a template or a sub-template $T_s$, if the counting of the sub-templates of $T_s$ has been completed, then the time complexity of counting $T_s$ is: 
\begin{equation}
O(|E|\binom{|T|}{|T_{s,p}|}  +|V| \binom{|T|}{|T_s|} \binom{|T_s|}{|T_{s,p}|}).
\end{equation}

\end{lemma}

\begin{proof}
Since $T_{s,p}$ has $\binom{|T|}{|T_{s,p}|}$ color combinations, which requires $\binom{|T|}{|T_{s,p}|}$ times SpMV operations. The total time consumption of this step is $O(|E|\binom{|T|}{|T_{s,p}|})$. 
In the next step,  $T_s$ has a total of $\binom{|T|}{|T_s|}$ different color combinations, and its sub-templates have a total of $\binom{|T_s|}{|T_{s,p}|}$  color combinations, a two-layer loop is needed here. The total time consumption of this step is $O(|V| \binom{|T|}{|T_s|} \binom{|T_s|}{|T_{s,p}|}) $. 
\end{proof}

\begin{lemma}\label{lem:ineq}
For integers $l \le m \le n$, and $l_0 \le 0, l_1 \le 1, \dots , l_n \le n,$ the following equations hold:

\begin{enumerate}
  \item $\max_{m} \binom{n}{m} = O(n^{-1/2}2^n)$
   \item $\max_{\{ m,l\}} \binom{n}{m} \binom{m}{l} = O(n^{-1}3^n)$
  \item $ \max_{\{ l_0, l_1,\dots ,l_n \}}\sum_{m=0}^{n}\binom{n}{m}\binom{m}{l_m} = O(n^{-1/2}3^n)$
\end{enumerate}
\end{lemma}
\begin{lemma}
In the worst case, the total time complexity of counting a $k$-node template using \ourimpl{} is 
\begin{equation}
O((e^k\log{(\frac{1}{\delta})} \frac{1}{\epsilon^2}) (|E| 2^k  + |V| k^{-1/2}3^k)).
\end{equation}
\end{lemma}
\begin{proof}
A template with $k$ nodes generates up to $O(k)$ sub-templates. 
And $O (e^k\log{(\frac{1}{\delta}})\frac{1}{\epsilon^2})$ iterations are performed in order to get the $(\epsilon,\delta)$-approximation. The conclusion is proved by combining Lemma~\ref{lem:ineq}. 
\end{proof}

\begin{table}[ht]
\centering
\caption{Comparison with Original Subgraph Counting}
\label{tab:theoretical_complexity}
\resizebox{\linewidth}{!}{%
\begin{tabular}{ccc} 
 \toprule
  \thead{\normalsize Time \\ \normalsize complexity} & One sub-template & One iteration\\ 
   \midrule
 \textsc{Fascia} & $O(|E| \binom{|T|}{|T_s|} \binom{|T_s|}{|T_{s,p}|})$ & $O(|E| \cdot k^{-1/2}3^{k})$ \\ 
 \textsc{Pfascia} & \thead{\normalsize $O(|E|\binom{|T|}{|T_{s,p}|} +$ \\ \normalsize $ |V| (\binom{|T|}{|T_s|} \binom{|T_s|}{|T_{s,p}|} +\binom{|T|}{|T_{s,p}|}))$} & $O(|E| 2^{k}  + |V| k^{-1/2}3^{k}$)  \\ 
 \textsc{Pgbsc} & \thead{\normalsize$O(|E|\binom{|T|}{|T_{s,p}|}+$\\\normalsize$|V| \binom{|T|}{|T_s|} \binom{|T_s|}{|T_{s,p}|})$ }& $O(|E|  2^{k} + |V|  k^{-1/2}3^{k})$  \\
 \bottomrule
\end{tabular}
}
\end{table}

\subsection{Estimation of Performance Improvement}
\label{sub:improve_estimation}

In this section, we present a model to estimate the run time of \ourimpl{} and its improvement over the original algorithm.

\paragraph{Run time}
From the previous analysis, we know that the execution time of \ourimpl{} on a single sub-template is $O(|E|\binom{|T|}{|T_{s,p}|}  +|V| \binom{|T|}{|T_s|} \binom{|T_s|}{|T_{s,p}|})$. By supplementing the constant term, we get 
\begin{equation}
 runtime_{PGBSC}= \alpha |E|\binom{|T|}{|T_{s,p}|}  +\beta |V| \binom{|T|}{|T_s|} \binom{|T_s|}{|T_{s,p}|},   
\end{equation}
where $\alpha$ and $\beta$ are constants, and we will calculate them by applying the data fitting. Similarly, we set:
\begin{equation}
 runtime_{FASCIA}= \gamma |E| \binom{|T|}{|T_s|} \binom{|T_s|}{|T_{s,p}|}.   
\end{equation}
Through the previous research we obtained that the time complexity of the original color-coding algorithm and \ourimpl{} are $O(|E|k^{-1/2}3^k)$ and $O(|E| 2^k  + |V| k^{-1/2}3^k)$ respectively. Since the $|E| 2^k$ term is very close to the actual value, and the $k^{-1/2}3^k$ term may be overestimated, to be more rigorous, we assume that the actual running time of the two programs is $|E|f(T)$ and $|E| 2^k  + |V|f(T)$ respectively. Here $T$ is the input template, and $f$ is a function for $T$. According to previous analysis, we obtain an upper bound of $f$: 
$ O(f(T)) \le O(k^{-1/2}3^k)$. The lower bound of $f(T)$ is estimated in this form: 
\begin{equation}
 f(T) = \sum_{T_s} {\binom{|T|}{|T_s|} \binom{|T_s|}{|T_{s}|} } \ge \sum_{T_s} \binom{|T|}{|T_s|}|T_s| \approx \sum_{i}i\binom{k}{i}=n\cdot 2^{k-1}   
\end{equation}
Thus, $O(k\cdot 2^k)\le O(f(T)) \le O(k^{-1/2}3^k)$. 
Comparing the complexity of the two algorithms we obtain: 
$Improvement \approx \frac{\gamma|E|f(T)}{\alpha |E| 2^k  + \beta |V| f(T)} = \frac{\gamma}{\alpha\frac{2^k}{f(T)} + \beta|V|/|E|} = \frac{\gamma}{\alpha\frac{2^k}{f(T)}  + \beta d^{-1}}.$


Since the influence of the constant term is not considered in the complexity analysis, we need to fill in the necessary constant terms: 

\begin{equation}
  \gamma (\alpha \cdot 2\cdot k^{-1} +\beta\cdot d^{-1})^{-1} \le Improvement \le \gamma(\alpha \cdot (\frac{2}{3})^k+\beta \cdot d^{-1})^{-1}, 
   \label{eq:speedup_bound}
\end{equation}
where $d$ is the average degree of the network, $\alpha$, $\beta$,and $\gamma$ are constants.

For a given graph, the following implications are obtained from this formula:
\begin{enumerate}
    \item The improvement grows with increasing template size, but no more than an upper bound, which is $\gamma\beta^{-1} d$.
    \item For a relatively small average degree $d$, the improvement is proportional to $d$, and the ratio is $\gamma\beta^{-1}$.
    \item For a relatively large average degree $d$, improvement will approach an upper bound between $\gamma\alpha^{-1}k$ and $\gamma\alpha^{-1}(\frac{3}{2})^{k} $.
\end{enumerate}

\section{Experimental Setup}
\label{sec:experiments}

\subsection{Software}
\label{sub:exp:impl}

We use the following three implementations. All of the binaries are compiled by the Intel(R) C++ compiler for Intel(R) 64 target platform from Intel(R) Parallel Studio XE 2019, with compilation flags
of ``-O3``, ``-xCore-AVX2", ``-xCore-AVX512", and the Intel(R) OpenMP. 
\begin{itemize}
    \item \textbf{\fasciaimpl{}} implements the FASCIA algorithm in~\cite{SlotaFastApproximateSubgraph2013}, which serves as a performance baseline.
    \item \textbf{\pvtsc{}} implements the data structure of \fasciaimpl{} but applying a pruning optimization on the graph traversal. 
    \item \textbf{\ourimpl{}} implements both of pruning optimization and 
    the GraphBLAS inspired vectorization. 
\end{itemize}

\subsection{Datasets and Templates}
\label{sub:datasets_and_templates}
\begin{figure}[ht]
    \centering
    \includegraphics[width=0.9\linewidth]{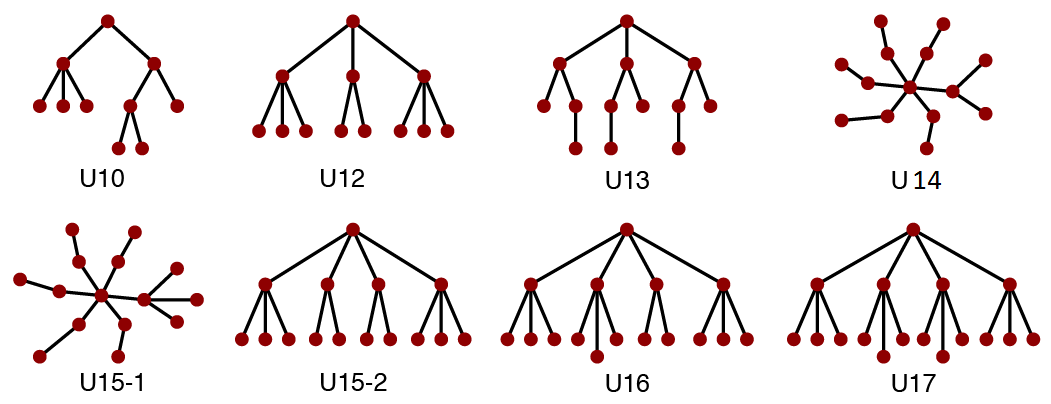}
    \caption{Templates used in experiments}
    \label{fig:templates}
\end{figure}
The datasets in our experiments are listed in Table~\ref{tab:datasets}, where \emph{Graph500 Scale=20, 21, 22} are collected from~\cite{graphchallengemit}; \emph{Miami}, \emph{Orkut}, and~\emph{NYC} are from ~\cite{barrett_generation_2009}~\cite{snapnets}~\cite{yang_defining_2012}; RMAT are widely used synthetic datasets generated by the RMAT model~\cite{chakrabarti_r-mat:_2004}. 
The templates in Figure~\ref{fig:templates} are from the tests in~\cite{SlotaFastApproximateSubgraph2013} or created by us. The template size increases from 10 to 17 while some templates have two different shapes.

\begin{table*}[ht]
\centering
\tiny
\caption{Datasets used in the experiments (K=$10^3$, M=$10^6$)}
\label{tab:datasets}
\resizebox{0.8\textwidth}{!}{%
    \begin{tabular}{lllllll}
        \toprule
        Data  &  Vertices &  Edges &   Avg Deg &   Max Deg & Abbreviation & Source\\ 
        \midrule
         Graph500 Scale=20  & 600K & 31M  & 48 & 67K & GS20 & Graph500~\cite{graphchallengemit}\\ 
         Graph500 Scale=21  & 1M & 63M  & 51 & 107K & GS21 & Graph500~\cite{graphchallengemit}\\ 
         Graph500 Scale=22   & 2M & 128M  & 53 & 170K & GS22 & Graph500~\cite{graphchallengemit}\\
        Miami &  2.1M & 200M &  49 &  10K & MI &  Social network~\cite{barrett_generation_2009}\\ 
        Orkut &  3M & 230M  & 76 & 33K  & OR &  Social network~\cite{snapnets}\\ 
        NYC   & 18M & 960M  & 54 & 429 & NY &  Social  network~\cite{yang_defining_2012}\\ 
        RMAT-1M & 1M & 200M & 201 & 47K  & RT1M & Synthetic data~\cite{chakrabarti_r-mat:_2004}\\
        RMAT(k=3) & 4M & 200M & 52 & 26K & RTK3 & Synthetic data~\cite{chakrabarti_r-mat:_2004}\\
        RMAT(k=5) & 4M & 200M & 73 & 144K & RTK5 & Synthetic data~\cite{chakrabarti_r-mat:_2004}\\
        RMAT(k=8) & 4M & 200M & 127 & 252K & RTK8 & Synthetic data~\cite{chakrabarti_r-mat:_2004}\\
        \bottomrule
    \end{tabular}
}
\end{table*}

\subsection{Hardware}
\label{sub:experimental setup}

In the experiments, we use: 1) a single node of a dual-socket Intel(R) Xeon(R) CPU E5-2670 v3 (architecture Haswell), 
and 2) a single node of a dual-socket Intel(R) Xeon(R) Platinum 8160 CPU (architecture Skylake-SP) processors.

\begin{table}[ht]
\centering

\caption{Node Specification of Testbed}
\label{tab:testbed}
\resizebox{\linewidth}{!}{%
\begin{tabular}{cccccc}
\toprule
Arch &  Sockets & Cores  &  Threads & \thead{ \normalsize CPU \\ \normalsize Freq} & \thead{\normalsize Peak \\ \normalsize Performance}\\
\midrule
Haswell & 2 & 24 & 48 & 2.3GHz & 1500 GFLOPS\\
Skylake & 2 & 48 & 96 & 2.1GHz & 4128 GFLOPS\\ 
\toprule
Arch & L1(i/d) & L2 & L3 & \thead{\normalsize Memory \\\normalsize  Size} & \thead{\normalsize Memory \\ \normalsize Bandwidth} \\ 
\midrule
Haswell & 32KB & 256KB & 30MB & 125GB & 95GB/s \\
Skylake & 32KB & 1024KB & 33MB & 250GB & 140GB/s \\
\bottomrule
\end{tabular}
}
\end{table}
The Operating System is Red Hat Enterprise Linux Server version 7.6 for both of the nodes, whose specifications are shown in Table~\ref{tab:testbed}. We use Haswell and Skylake to refer to the Intel(R) Xeon(R) CPU E5-2670 v3 and Intel(R) Xeon(R) Platinum 8160 CPU respectively in the rest of the paper. The memory bandwidth is measured by STREAME Benchmark~\cite{McCalpin1995} while the peak performance is measured by using the Intel(R) Optimized LINPACK Benchmark for Linux. We use, by default, the number of threads equal to the number of physical cores, i.e., 48 threads on Skylake node and 24 threads on Haswell node. The threads are bound to cores with a spread affinity. As the Skylake node has twice the memory size and physical cores as the Haswell node, we use it as our primary testbed in the experiments. 

\section{Performance and Hardware Utilization}
\label{sec:results_and_analysis}

\begin{figure*}[ht]
    \centering
    \subfloat[Miami]{\includegraphics[width=0.33\textwidth]{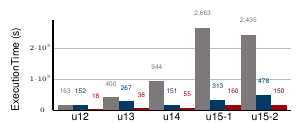}}
    \subfloat[Orkut]{\includegraphics[width=0.33\textwidth]{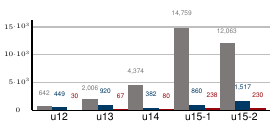}}
    \subfloat[RMAT-1M]{\includegraphics[width=0.33\textwidth]{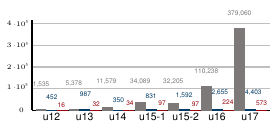}}
    \vspace{-3ex}
    \subfloat[Graph500 Scale=20]{\includegraphics[width=0.33\textwidth]{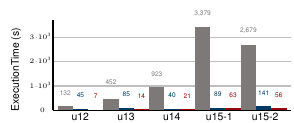}}
    \subfloat[Graph500 Scale=21]{\includegraphics[width=0.33\textwidth]{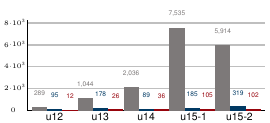}}
    \subfloat[Graph500 Scale=22]{\includegraphics[width=0.33\textwidth]{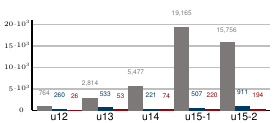}}
   \vspace{-3ex}
    \subfloat[RMAT K=3]{ \includegraphics[width=0.33\textwidth]{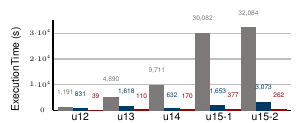}}
    \subfloat[RMAT K=5]{\includegraphics[width=0.33\textwidth]{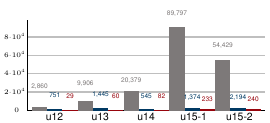}}
    \subfloat[RMAT K=8]{\includegraphics[width=0.33\textwidth]{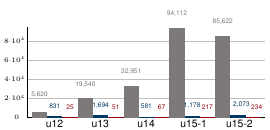}}
    \vfill
    \includegraphics[width=0.25\textwidth]{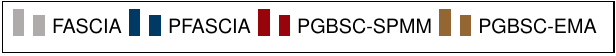}
    \caption{Execution time of \fasciaimpl{} ~versus \pvtsc{} versus \ourimpl{} with increasing template sizes, from U12 to U17. Tests run
    on a Skylake node.}
    \label{fig:ExecutionTime}
\end{figure*}
 
In Figure~\ref{fig:ExecutionTime}(a)(b)(c), we scale the template size up to the memory limitation of our testbed for each dataset, and the reduction of execution time is significant particularly for template sizes larger than 14. 
For instance, \fasciaimpl{} spends four days to process a million-vertex dataset RMAT-1M with template u17 while \ourimpl{} only spends 9.5 minutes. For relatively smaller templates such as u12, \ourimpl{} still achieves 10x to 100x the reduction of execution time on datasets Miami, Orkut, and RMAT-1M. 
As discussed in Section~\ref{sub:improve_estimation}, the improvement is proportional to its average degree. 
This is observed when comparing datasets Miami (average degree of 49) and Orkut (average degree of 76), where \ourimpl{} achieve 10x and 20x improvement on u12 respectively. 
In Figure~\ref{fig:ExecutionTime}(d)(e)(f), we scale up the size of Graph500 datasets in Table~\ref{tab:datasets}. Note that all of the Graph500 datasets have similar average degrees and get similar improvement by \ourimpl{} for the same template. This implies that \ourimpl{} delivers the same performance boost to datasets with a comparable average degree but different scales. Finally, in Figure~\ref{fig:ExecutionTime}(g)(h)(i), we compare RMAT datasets with increasing skewness, which causes the imbalance of vertex degree distribution. The results show that \ourimpl{} has comparable execution time regardless of the growing degree distribution skewness. In contrast, \fasciaimpl{} spends significantly (2x to 3x) more time on skewed datasets. 
To validate the performance gained by using \ourimpl{}, we investigate the performance improvement contributed by pruning and vectorization, accordingly to Section~\ref{sub:exp_prune} and~\ref{sub:exp_vectorization}.

\subsection{Pruning Optimization}
\label{sub:exp_prune}

In Figure~\ref{fig:fascia-pfascia-increasetemp}, we compare six datasets with each starting from the template u12 and up to the largest templates whose count data the Skylake node can hold.
Three observations are made: 1) \pvtsc{} achieves a performance improvement of 10x by average and up to 86x. 
2) \pvtsc{} obtains higher relative performance on large
templates. For instance, it gets 17.2x 
improvement for u15-1 while only 2.2x for u13 for dataset Orkut. 
3) \pvtsc{} works better at datasets with high skewness of degree distribution. Datasets like Miami and Orkut that have a power law distribution only get 8.5x and 17.2x improvement at u15-1, while the RMAT-1M obtains 41x improvement. 

\begin{figure}[ht]
    \centering
    \includegraphics[width=\linewidth]{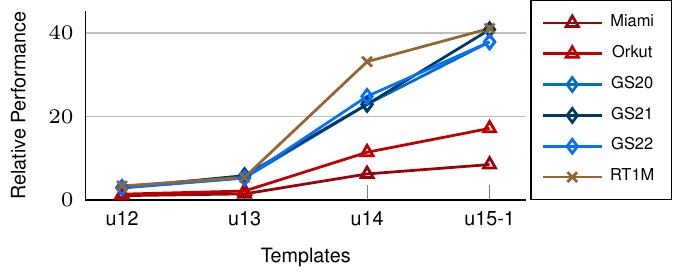}
    \caption{Relative performance of \pvtsc{} versus \fasciaimpl{} with increasing template sizes. Tests are done on a Skylake node.}
    \label{fig:fascia-pfascia-increasetemp}
\end{figure}

\subsection{Vectorization Optimization}
\label{sub:exp_vectorization}

We decompose \ourimpl{} into the two linear algebra kernels SpMM and eMA as described in Section~\ref{sub:kernel-spmm-ema}.

\subsubsection{\textbf{CPU and VPU Utilization}}
\label{sub:cpu-vpu-usage}

Figure~\ref{fig:cpu-vpu-perf}(a) first compares the CPU utilization defined as the average number of concurrently running physical cores. 
For Miami, \pvtsc{} and \fasciaimpl{} achieve 60\% and 75\% of CPU utilization. However, the CPU utilization drops below 50\% on Orkut and NYC who have more vertices and edges than
Miami. Conversely, SpMM kernel keeps a high ratio of 65\% to 78\% for all three of the datasets, and the eMA kernel has a growing CPU utilization from Miami (46\%) to NYC (88\%). We have two explanations: 1) the SpMM kernel splits and regroups the nonzero entries by their row IDs, which mitigates the imbalance of nonzero entries among rows; 2) the eMA kernel has its computation workload for each column of $M_{s,a}, M_{s,p}$ evenly dispatched to threads. \par
\begin{figure}[ht]
    \centering
    \subfloat[CPU Utilization]{    \includegraphics[width=0.9\linewidth]{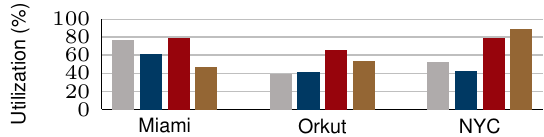}}
    \vspace{-3ex}
    \subfloat[VPU Utilization]{\includegraphics[width=0.9\linewidth]{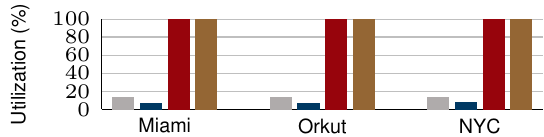}}
     \vspace{-3ex}
    \subfloat[Packed FP]{
    \includegraphics[width=0.9\linewidth]{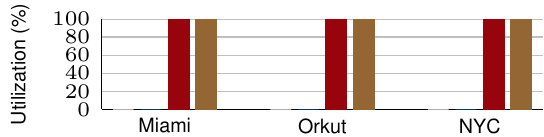}}
    \vfill
    \includegraphics[width=0.4\textwidth]{legend.pdf}
    \caption{The hardware utilization on Skylake node for 
    template u12.}
    \label{fig:cpu-vpu-perf}
\end{figure}
Secondly, we examine the code vectorization in Figure~\ref{fig:cpu-vpu-perf}. VPU in a Skylake node is a group of 512-bit registers. The scalar instruction also utilizes the VPU but it cannot fully exploit its 512-bit length. 
Figure~\ref{fig:cpu-vpu-perf} refers to the portion of instructions vectorized with full vector capacity. 
For all of the three datasets, \pvtsc{} and \fasciaimpl{} only have 6.7\% and 12.5\% VPU utilization implying that the codes may not be vectorized, while for SpMM and eMA kernels of \ourimpl{}, the VPU utilization is 100\%. 
A further metric of packed float point instruction ratio (Packed FP) justifies the implication that \pvtsc{} and \fasciaimpl{} have zero vectorized instructions but \ourimpl{} has all of its float point operations vectorized. 

\subsubsection{\textbf{Memory Bandwidth and Cache}}
\label{sub:bd-cache}

Because of the sparsity, subgraph counting is memory-bounded in a shared memory system. Therefore, the utilization of memory and cache resources are critical to the overall performance. In Table~\ref{tab:exp:mem-cache}, we compare SpMM and eMA of to \pvtsc{} and \fasciaimpl{}. It shows that the eMA kernel has the highest bandwidth value around 110 GB/s for the three datasets, which is due to the highly vectorized codes and regular memory access pattern. Therefore, the data is prefetched into cache lines, which controls the cache miss rate as low as 0.1\%. \par 
The SpMM kernel also enjoys a decent bandwidth usage around 70 to 80 GB/s by average when compared to \pvtsc{} and \fasciaimpl{}. Although SpMM has an L3 miss rate as high as 74\% in dataset NYC because the dense matrix is larger than the L3 cache capacity. The optimized thread level and instruction level vectorization ensures a concurrent data loading from memory leveraging the high memory bandwidth utilization. 
\fasciaimpl{} has the lowest memory bandwidth usage because of the thread imbalance and the irregular memory access. 
\begin{table}[ht]
    \centering
    \caption{Memory and Cache Usage on Skylake Node}
    \resizebox{\linewidth}{!}{
    \begin{tabular}{lllll}
    \toprule
    Miami & Bandwidth & L1 Miss Rate & L2 Miss Rate & L3 Miss Rate  \\
    \midrule
    FASCIA & 6 GB/s & 4.1\% & 1.8\% & 85\% \\ 
    PFASCIA & 48.8 GB/s & 9.5\% & 86.5\% & 50\% \\ 
    PGBSC-SpMM & 86.95 GB/s & 8.3\% & 51.2\% & 36.8\% \\ 
    PGBSC-eMA & 106 GB/s & 0.3\% & 20.6\% & 9.9\% \\ 
    \midrule
    Orkut & Bandwidth & L1 Miss Rate & L2 Miss Rate & L3 Miss Rate  \\
    \midrule
    FASCIA & 8 GB/s & 9.6\% & 5.3\% & 46\% \\ 
    PFASCIA & 30 GB/s & 8.4\% & 76.2\% & 46\% \\ 
    PGBSC-SpMM & 59.5 GB/s & 6.7\% & 42.8\% & 45\% \\ 
    PGBSC-eMA & 116 GB/s & 0.32\% & 22.2\% & 9.0\% \\ 
    \midrule
    NYC & Bandwidth & L1 Miss Rate & L2 Miss Rate & L3 Miss Rate  \\
    \midrule
    FASCIA &  7 GB/s & 2.4\% & 8.1\% & 87\% \\ 
    PFASCIA & 38 GB/s & 10\% & 90\% & 81\% \\ 
    PGBSC-SpMM & 96 GB/s & 7.7\% & 76\% & 74\% \\ 
    PGBSC-eMA & 122 GB/s & 0.1\% & 99\% & 14.8\% \\ 
    \bottomrule
    \end{tabular}}
    \label{tab:exp:mem-cache}
\end{table}
\subsubsection{\textbf{Roofline Model}}

The roofline model in Figure~\ref{fig:roofline} reflects
the hardware efficiency. The horizontal axis is the
operational intensity (FLOP/byte) and the vertical axis 
refers to the measured throughput performance (FLOP/second). 
The solid roofline is the maximal performance the hardware can deliver under a certain operational intensity.
From \fasciaimpl{} to \pvtsc{},
the performance gap to the roofline grows, implying that 
the pruning optimization itself does not improve the hardware utilization although it removes redundant computation. 
From \pvtsc{} to \ourimpl{}, the performance gap to the roofline is reduced significantly, resulting in a performance point hit by the roofline. This near-full hardware efficiency is contributed by the vectorization optimization, which exploits the memory bandwidth usage.\par
\begin{figure}[ht]
 \centering
 \includegraphics[width=\linewidth]{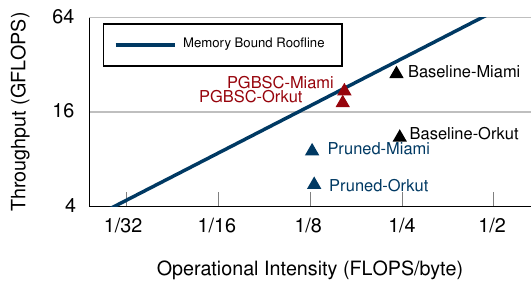}
 \caption{
 Apply roofline model to \fasciaimpl{}, \pvtsc{},and 
 \ourimpl{}. Dataset Miami, Orkut for template u15-1. Tests
 are done on a Skylake node.
}
\label{fig:roofline}
\end{figure}
Figure ~\ref{fig:hsw-skl} compares the performance of \ourimpl{} between the Skylake node and the Haswell node. For both test beds, we set up threads numbered equal to their physical cores, and Skylake node in has a 1.7x improvement over Haswell node. 
Although the Skylake node doubles the CPU cores compared to the Haswell node, it increases the peak memory bandwidth by only 47\% in Table~\ref{tab:testbed}. 
As \ourimpl{} is bounded by the memory roofline in
Figure~\ref{fig:roofline}, the estimated improvement shall be a value between 1.47x and 2x, and the $1.7x$ improvement by \ourimpl{} is reasonable.
 \begin{figure}[ht]
    \centering
    \includegraphics[width=\linewidth]{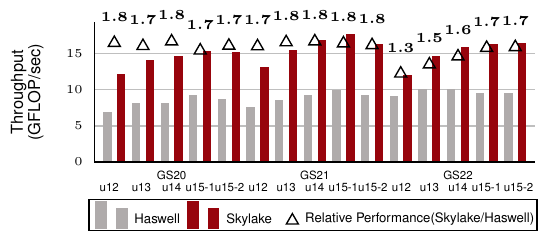}
    \caption{The performance throughput of \ourimpl{} on 
    Haswell node versus. Skylake node.}
    \label{fig:hsw-skl}
\end{figure}

\subsection{Load Balance and Thread Scaling}
\begin{figure}[ht]
 \centering
 \includegraphics[width=\linewidth]{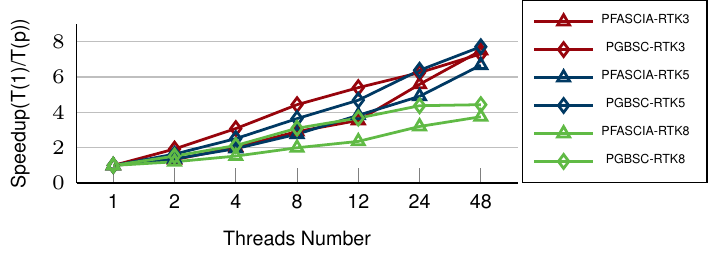}
 \caption{Thread scaling for RMAT datasets with increasing skewness on a Skylake node.} 
 \label{fig:thdscale}
\end{figure}
We perform a strong scaling test using up to 48 threads on Skylake node in Figure~\ref{fig:thdscale}. We choose RMAT generated datasets with increasing skewness parameters of $K=3, 5, 8$. When $K=3$, \ourimpl{} has a better thread scaling than
\pvtsc{} from 1 to 12 threads; after 12 threads, 
the thread scaling of \ourimpl{} slows down.
As the performance is bounded by memory, which has 6 memory
channels per socket, we have a total of 12 memory
channels on a Skylake node that bounds the thread scaling. 
To the contrary, \pvtsc{} is not bounded by the memory
bandwidth, which explains why it keeps an 
increasing thread scaling from 12 to 48 threads. 
Eventually, both of \ourimpl{} and \pvtsc{} 
obtain a 7.5x speedup at 48 threads. 
When increasing the skewness of datasets to $K=5,8$, 
the thread scalability of \pvtsc{} and \ourimpl{} both drop 
down because the skewed data distribution brings workload
imbalance when looping vertex neighbors. 
However, \ourimpl{} has a better scalability 
than \pvtsc{} at 48 threads because of the SpMM kernel.

\subsection{Error Discussion}
\ourimpl{} with its pruning and vectorization optimization only differs from the \fasciaimpl{} due to the restructuring of the computation from Algorithm 1 to Algorithm 4, and so should give identical results with exact arithmetic in Equation~\ref{eq:transform_traversal}. However,the range of values needed when processing large templates. As a consequence, both \fasciaimpl{} and our \ourimpl{} use floating point numbers to avoid overflow. Hence, slightly different results are observed between \fasciaimpl{} and \ourimpl{} due to the rounding error consequent from floating point arithmetic operations. Figure~\ref{fig:error} reports such relative error in the range of $10^{-6}$ across all the tests on a Graph500 GS20 dataset with increasing template sizes, which is negligible.
\begin{figure}[ht]
    \centering
    \includegraphics[width=\linewidth]{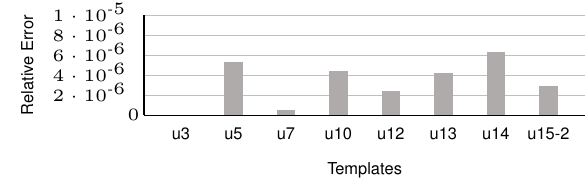}
     \caption{Relative error on dataset Graph500 Scale=20.
     Tests are done on a Skylake node.}
    \label{fig:error}
\end{figure}
\subsection{Overall Performance Improvement}
Figure~\ref{fig:fascia-pgbsc-increasetemp} shows significant performance improvement of \ourimpl{} over \fasciaimpl{} for a variety of networks and subtemplates; 
the relative performance grows with template sizes. 

\begin{figure}[ht]
    \centering
    \includegraphics[width=\linewidth]{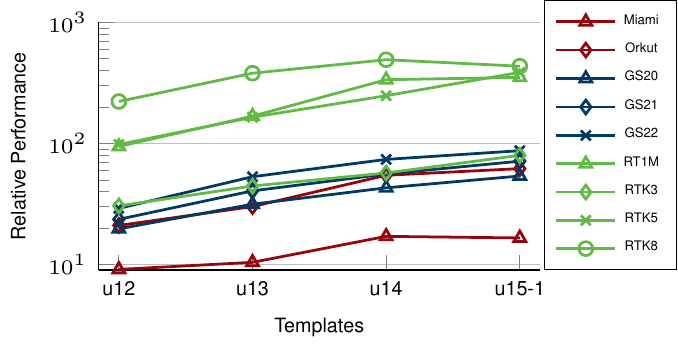}
     \caption{Performance improvement of \ourimpl{} vs. \fasciaimpl{} with increasing template sizes. 
     Tests are done on a Skylake node.}
    \label{fig:fascia-pgbsc-increasetemp}
\end{figure}

For small dataset Miami and template u12, they still 
achieve 9x improvement, and for datasets like Graph500
(scale:20) and templates with size starting from u14,
\ourimpl{} achieves around 50x improvement by average and 
up to 660x improvement for synthetic dataset RT1M at large
template u17.

\section{Conclusion}
\label{sec:conclusion}

As the single machine with big shared memory and many cores is becoming an attractive solution to graph analysis problems~\cite{perez2015ringo}, the irregularity of memory access remains a 
roadblock to improve hardware utilization.  
For fundamental algorithms, such as PageRank, the fixed data structure and predictable execution order are explored to improve data locality either in graph traversal approach ~\cite{lakhotia2018accelerating}\cite{zhang2017making} 
or in linear algebra approach~\cite{beamer2017reducing}. Subgraph counting, with random access to vast memory region and dynamic programming workflow, requires much more effort to exploit the cache efficiency and hardware vectorization.
To the best of our knowledge, we are the first to fully vectorize a 
sophisticated algorithm of subgraph analysis, and the novelty is a co-design approach to combine algorithmic improvement with pattern identification of linear algebra kernels that leverage hardware vectorization. \par

The overall performance achieves a promising improvement over 
the state-of-the-art work by orders of magnitude by average 
and up to 660x (RMAT1M with u17) within a shared-memory 
multi-threaded system, and we are confident 
to generalize this GraphBLAS inspired approach in 
our future work: 1) enabling counting of tree subgraph 
with size larger than 30 and subgraph beyond trees; 
2) extending the shared-memory implementation to distributed system 
upon our prior work~\cite{chen2018high}~\cite{SlotaComplexNetworkAnalysis2014a}
~\cite{zhao2018finding}~\cite{zhao_sahad:_2012}; 3) exploring 
other graph and machine learning problems by this GraphBLAS inspired co-design 
approach; 4) adding support to more
hardware architectures (e.g, NEC SX-Aurora and NVIDIA GPU).
The codebase of our work on \ourimpl{} is made public in our open-sourced repository\footnote{\url{https://github.com/DSC-SPIDAL/harp/tree/pgbsc}}.
\section{Acknowledgments}
We gratefully acknowledge generous support from the Intel Parallel Computing Center (IPCC),  the NSF CIF-DIBBS 1443054: Middleware and High-Performance Analytics Libraries for Scalable Data Science, and NSF BIGDATA 1838083: Enabling Large-Scale, Privacy-Preserving Genomic Computing with a Hardware-Assisted Secure Big-Data Analytics Framework. We appreciate the support from the IU PHI grand challenge, the FutureSystems team, and the ISE Modelling and Simulation Lab.

\bibliographystyle{ACM-Reference-Format}
\bibliography{refs}

\end{document}